\documentclass[11pt]{article}
\usepackage[margin=1in]{geometry}
\usepackage{datetime}
\usepackage{enumitem}
\usepackage[T1]{fontenc}
\usepackage[utf8]{inputenc}
\usepackage{pgfplots}
\usepackage{amsmath, amssymb}
\usepackage{amsthm}
\usepackage{color}
\usepackage[numbers]{natbib}
\usepackage{cases}
\usepackage{xparse}
\usepackage{xargs}
\usepackage{appendix}
\usepackage{algorithm}
\usepackage{algpseudocode}
\usepackage{verbatim, xspace}
\usepackage{tikz}
\usetikzlibrary{arrows}
\usetikzlibrary{patterns}
\usetikzlibrary{calc}
\usetikzlibrary{decorations.pathreplacing,angles,quotes}
\usepackage{hyperref}
\usepackage{pgfplots}

\usepackage{thmtools}
\usepackage{thm-restate}
\usepackage{cleveref}
\usepackage{xcolor}
\usepackage[disable,colorinlistoftodos,prependcaption,textsize=tiny]{todonotes}
\newcommandx{\unsure}[2][1=]{\todo[linecolor=green,backgroundcolor=green!25,bordercolor=green,#1]{\normalsize #2}}
\newcommandx{\improvement}[2][1=]{\todo[inline,linecolor=blue,backgroundcolor=blue!05,bordercolor=blue,#1]{\normalsize #2}}
\newcommandx{\info}[2][1=]{\todo[linecolor=yellow,backgroundcolor=yellow!25,bordercolor=yellow,#1]{#2}}
\newcommandx{\floatmodel}[2][1=]{\todo[inline,linecolor=red,backgroundcolor=yellow!25,bordercolor=yellow,#1]{#2}}
\newcommandx{\thiswillnotshow}[2][1=]{\todo[disable,#1]{#2}}
\newcommandx{\karolw}[2][1=]{\todo[inline,linecolor=green,backgroundcolor=green!25,bordercolor=green,#1]{\normalsize #2}}

\newtheorem{theorem}{Theorem}
\newtheorem{definition}[theorem]{Definition}
\newtheorem{lemma}[theorem]{Lemma}

\newtheorem{claim}[theorem]{Claim}

\newtheorem{problem}{Open Question}

\numberwithin{theorem}{section}
\numberwithin{lemma}{section}
\numberwithin{claim}{section}
\numberwithin{corollary}{section}
\numberwithin{definition}{section}

\newcommand{\floor}[1]{\left\lfloor #1 \right\rfloor}
\newcommand{\ceil}[1]{\left\lceil #1 \right\rceil}

\newcommand{\eps}{\varepsilon}

\newcommand{\Oh}{\mathcal{O}}
\newcommand{\Os}{\mathcal{O}^{*}}
\newcommand{\Otilde}{\widetilde{\mathcal{O}}}
\newcommand{\Ot}{\Otilde}

\newcommand{\poly}{\mathrm{poly}}
\newcommand{\polylog}{\,\textup{polylog}}
\newcommand{\Mod}[1]{\ (\mathrm{mod}\ #1)}
\newcommand{\Sgn}[1]{\mathrm{sgn}(#1)}
\newcommand{\prob}[2]{\mathbb{P}_{#2}\left[ #1 \right]}
\newcommand{\Ex}[1]{\mathbb{E}\left[ #1 \right]}

\newcommand{\Ess}{\textsc{Equal-Subset-Sum}\xspace}
\newcommand{\Ss}{\textsc{Subset-Sum}\xspace}

\title{Equal-Subset-Sum Faster Than the Meet-in-the-Middle}
\date{}

\author{
    Marcin Mucha\footnote{Institute of Informatics, University of
    Warsaw, Poland, \texttt{mucha@mimuw.edu.pl}. Supported by project TOTAL that has received funding from the European
    Research Council (ERC) under the European Union’s Horizon 2020 research and
    innovation programme (grant agreement No 677651).}
    \and
    Jesper Nederlof\footnote{Eindhoven University of Technology, The
    Netherlands, \texttt{j.nederlof@tue.nl}. Supported by the Netherlands
    Organization for Scientific Research under project no. 024.002.003 and the
    European Research Council under project no. 617951.}
    \and
    Jakub Pawlewicz\footnote{Institute of Informatics, University of
    Warsaw, Poland, \texttt{pan@mimuw.edu.pl}. }
    \and
    Karol W\k{e}grzycki\footnote{Institute of Informatics, University of
    Warsaw, Poland, \texttt{k.wegrzycki@mimuw.edu.pl}. Supported by
    the grants 2016/21/N/ST6/01468 and 2018/28/T/ST6/00084 of the Polish National
    Science Center and project TOTAL that has received funding from the European
    Research Council (ERC) under the European Union’s Horizon 2020 research and
    innovation programme (grant agreement No 677651).}
}

\begin{document}

\maketitle

\thispagestyle{empty}
In the \Ess problem, we are given a set $S$ of $n$ integers and the problem is to
decide if there exist two disjoint nonempty subsets $A,B \subseteq S$, whose
elements sum up to the same value. The problem is NP-complete. The
state-of-the-art algorithm runs in $\Os(3^{n/2}) \le \Os(1.7321^n)$ time and is based on
the \emph{meet-in-the-middle} technique.  In this paper, we improve upon this
algorithm and give $\Os(1.7088^n)$ worst case Monte Carlo algorithm. 
This answers a question suggested by Woeginger in his inspirational survey. 

Additionally, we analyse the polynomial space algorithm for \Ess. A naive
polynomial space algorithm for \Ess runs in $\Os(3^n)$ time. With read-only
access to the exponentially many random bits, we show a randomized algorithm
running in $\Os(2.6817^n)$ time and polynomial space.
 
\clearpage
\setcounter{page}{1}

\section{Introduction}

In the \Ss problem, we are given as input a set $S$ of $n$ integers
$a_1,\ldots,a_n$ and a target $t$. The task is to decide if there exists a subset of $S$,
such that a total sum of the numbers in this subset is equal to $t$. This can
be formulated in the following form:

\begin{displaymath}
    \sum_{i=1}^n x_i a_i = t
\end{displaymath}

and the task is to find $x_i \in \{0,1\}$. \Ss is one of the fundamental NP-complete
problems. Study on the exact complexity of \Ss led to the discovery of one of the most fundamental
algorithmic tool: \emph{meet-in-the-middle}. \citet{horowitz-sahni} used
this technique to give a
$\Os(2^{n/2})$ algorithm for \Ss in the following way: First, rewrite the \Ss equation:

\begin{displaymath}
    \sum_{i=1}^{\floor{n/2}} x_i a_i = t - \sum _{\floor{n/2}+1}^n x_i a_i 
    .
\end{displaymath}

Then enumerate all $\Oh(2^{n/2})$ possible values of the left side
$L(x_1,\ldots,x_{\floor{n/2}})$ and
$\Oh(2^{n/2})$ possible values of the right side $R(x_{\floor{n/2}+1},\ldots,x_n)$. After that, it remains to look
for the value that occurs in both $L$ and $R$, i.e., \emph{meeting} the tables
$L$ and $R$. One can do that efficiently by sorting (see~\cite{horowitz-sahni} for
details). To summarize, meet-in-the-middle technique is based on rewriting the
formula as an equation between two functions and efficiently seeking any
value that occurs in both of their images.

Later, \citet{schroeppel} observed that space usage of meet-in-the-middle can
be improved to $\Os(2^{n/4})$ by using space-efficient algorithm for 4-SUM.
However, the time complexity remains unchallenged and 
\todo{MM: I usually try to avoid such unqualified statements, maybe most prominent
in exact algos, but probably not in general? KW: done}
one of the most prominent open problem in the area of exact algorithms is to improve upon \emph{meet-in-the-middle} for \Ss:

\begin{problem}
    Can \Ss be solved in $\Os(2^{(0.5 - \delta) n})$ time for some constant
    $\delta > 0$?
\end{problem}

In this paper, we consider the \Ess problem. We are given a
set $S$ of $n$ integers and the task is to decide if there exist two disjoint
nonempty subsets $A,B \subseteq S$, whose elements sum up to the same value.
Similarly to \Ss, this problem is NP-complete~\cite{woeginger92}. In the inspirational survey,
\citet{woeginger-survey2} noticed \Ess can be solved by using meet-in-the-middle and
asked if it can be improved: \footnote{\cite{woeginger-survey1,woeginger-survey2} noticed that
4-SUM gives $\Os(2^n)$ algorithm, but it actually gives a $\Os(3^{n/2})$
algorithm, see Appendix~\ref{table-4-sum-ess}.}

\begin{problem}[c.f.,~\cite{woeginger-survey1},\cite{woeginger-survey2}]
    \label{main-problem}
    Can we improve upon the meet-in-the-middle algorithm for \Ess?
\end{problem}

The folklore meet-in-the-middle algorithm for \Ess (that we will present in the
next paragraph) works in $\Os(3^{n/2})$ time.

\karolw{Do we want to present this algorithm or ref to Appendix is enough? JN: I think Appendix is enough, but can also keep as is.}

\paragraph{Folklore algorithm for \Ess} First, we arbitrarily partition $S$ into $S_1 =
\{a_1,\ldots,a_{\floor{n/2}}\}$ and $S_2 = \{a_{\floor{n/2}+1},\ldots,a_n\}$. Recall that in
\Ess we seek two subsets $A,B \subseteq S$, such that $A \cap B =
\emptyset$ and $\Sigma(A) = \Sigma(B)$. We can write the solution as 4 subsets: $A_1 = A
\cap S_1$, $A_2 = A \cap S_2$, $B_1 = B \cap S_1$ and $B_2 = B \cap S_2$, such
that: $\Sigma(A_1) + \Sigma(A_2) = \Sigma(B_1) + \Sigma(B_2)$. In particular, it means
that: $\Sigma(A_1) - \Sigma(B_1) = \Sigma(B_2) - \Sigma(A_2)$. So, the problem reduces
to finding two vectors $x \in \{-1,0,1\}^{\floor{n/2}}$ and $y \in \{-1,0,1\}^{\ceil{n/2}}$, such that:

\begin{displaymath}
    \sum_{i=1}^{\floor{n/2}} x_i a_i = \sum_{i=1}^{\ceil{n/2}} y_i a_{i+\floor{n/2}}
    .
\end{displaymath}

We can do this in $\Os(3^{n/2})$ time as follows. First, enumerate and store all
$3^{\floor{n/2}}$ possible values of the left side of the equation and all $3^{\ceil{n/2}}$
possible values of the right side of the equation. Then look for a value that
occurs in both tables (collision)  in time $\Os(3^{n/2})$ by sorting the values. The total running time is therefore
$\Os(3^{n/2})$. Analogously to \Ss, one can improve the space usage of the above algorithm to
$\Os(3^{n/4})$ (see Appendix~\ref{table-4-sum-ess}).

A common pattern seems unavoidable in algorithms for \Ss and \Ess: we have
to go through all possible values of the left and the right side of the equation. 
This enumeration dominates the time used to solve the problem. So, it was
conceivable that perhaps no improvement for \Ess could be obtained unless we
improve an algorithm for \Ss first~\cite{woeginger-survey1,woeginger-survey2}.
 \subsection{Our Contribution}

While the meet-in-the-middle algorithm remains unchallenged for \Ss, we show that, surprisingly, 
 we can improve the algorithm for \Ess. The
main result of this paper is the following theorem.

\begin{restatable}{theorem}{expspacethm}
    \label{expspace-alg}
    \Ess can be solved in $\Os(1.7088^n)$ time with high probability.
\end{restatable}

This positively answers Open Question~\ref{main-problem}. To
prove this result we observe that the worst case for the meet-in-the-middle algorithm
is that of a balanced solution, i.e., when $|A|=|B|=|S\setminus (A\cup B)|\approx n/3$.  
We propose a substantially different algorithm, that runs in $\Os(2^{2/3n})$ time for that case.  
The crucial insight of the new approach is the fact that when $|A|\approx |B|\approx
n/3$, then there is an abundance of pairs $X,Y \subseteq S$, $X \neq Y$ with $\Sigma(X)=\Sigma(Y)$. 
We use the \emph{representation technique} to exploit this. Interestingly, that
technique was initially developed to solve the average case
\Ss~\cite{generic-knapsack2,generic-knapsack1}.

Our second result is an improved algorithm for \Ess running in polynomial space.
The naive algorithm in polynomial space works in $\Os(3^n)$ time by
enumerating all possible disjoint pairs of subsets of $S$. This algorithm is
analogous to the $\Os(2^n)$ polynomial space algorithm for \Ss. Recently,
\citet{polyspace-stoc2017} proposed a $\Os(2^{0.86n})$ algorithm for \Ss on the
machine that has access to the exponential number of random bits. We show that
a similar idea can be used for \Ess.

\begin{restatable}{theorem}{polyspacethm}
    \label{improved-polyspace}
    There exists a Monte Carlo algorithm which solves \Ess in polynomial space
    and time $\Os(2.6817^n)$. The algorithm assumes random read-only access to
    exponentially many random bits.
\end{restatable}

This result is interesting for two reasons. First, \citet{polyspace-stoc2017}
require nontrivial results in information theory. Our algorithm is relatively
simple and does not need such techniques. Second, the approach of
\citet{polyspace-stoc2017} developed for \Ss has a barrier, i.e., significantly
new ideas must be introduced to get an algorithm running faster than
$\Os(2^{0.75n})$. In our case, this corresponds to the algorithm running in
$\Os(2^{1.5n}) \le \Os(2.8285^n)$ time and polynomial space (for elaboration see
Section~\ref{polyspacealg:sec}). We show that relatively simple observations about
\Ess enable us to give a slightly faster algorithm in polynomial space.
 \subsection{Related Work}

The \Ess was introduced by \citet{woeginger92} who showed that the problem is
NP-complete. This reduction automatically excludes $2^{o(n)}$ algorithms for
\Ess assuming ETH (see Appendix~\ref{conditional-lower-bound}), hence for this
problem we aspire to optimize the constant in the exponent.
The best known constant comes from the meet-in-the-middle algorithm. 
\citet{woeginger-survey2} asked if this algorithm for \Ess can be
improved. 

\paragraph{Exact algorithms for \Ss:}

\citet{nederlof-mfcs2012} proved that in the exact setting Knapsack and \Ss problems are equivalent. 

\citet{schroeppel} showed that the meet-in-the-middle algorithm admits a time-space
tradeoff, i.e., $\mathcal{T}\mathcal{S}^2 \le \Os(2^n)$, where $\mathcal{T}$ is
the running time of the algorithm and $\mathcal{S} \le \Os(2^{n/2})$ is the
space of an algorithm.  This tradeoff was improved by \citet{subsetsum-tradeoff}
for almost all tradeoff parameters.

\citet{stacs2015} considered \Ss parametrized by the maximum bin size $\beta$ and
obtained algorithm running in time $\Os(2^{0.3399n}\beta^4)$. Subsequently,
\citet{stacs2016} showed that one can get a faster algorithm for \Ss than
meet-in-the-middle if $\beta \le 2^{(0.5 - \eps)n}$ or $\beta \ge 2^{0.661n}$.
\todo{MM: derivation of their ideas?: KW done}
In this paper, we use the hash function that is based on their ideas. Moreover, the ideas in
\cite{stacs2015,stacs2016} were used in the recent breakthrough polynomial space
algorithm~\cite{polyspace-stoc2017} running in $\Os(2^{0.86n})$ time.

From the pseudopolynomial algorithms perspective Knapsack and \Ss admit
$\Oh(nt)$ algorithm, where $t$ is a value of a target. Recently, for \Ss the
pseudopolynomial algorithm was
improved to run in deterministic $\Ot(\sqrt{n}t)$ time by \citet{koiliaris-soda} and
randomized $\Ot(n+t)$ time by \citet{bringmann-soda} (and simplified, see
\cite{sosa-wu,simplified-koiliaris}). However, these algorithms have a drawback of
running in pseudopolynomial space $\Os(t)$. Surprisingly,
\citet{pseudopolynomial-polyspace} presented an algorithm running in time $\Ot(n^3t)$
and space $\Ot(n^2)$ which was later improved to $\Ot(nt)$ time and $\Ot(n
\log{t})$ space assuming the Extended Riemann Hypothesis~\cite{bringmann-soda}.

From a lower bounds perspective, no algorithm
working in $\Ot(\text{poly}(n) t^{0.99})$ exists for \Ss assuming SETH or SetCover
conjecture~\cite{subset-sum-lower,subset-sum-lower2}.

\paragraph{Approximation:}

\citet{woeginger92} presented the approximation algorithm for \Ess with the worst
case ratio of $1.324$. \citet{bazgan1998efficient} considered a different
formulation of approximation for \Ess and showed an FPTAS for it.

\paragraph{Cryptography and the average case complexity:}

In 1978 Knapsack problems were introduced into cryptography by
\citet{crypto-merkle78}. They introduced a Knapsack based public key
cryptosystem. Subsequently, their scheme was broken by using lattice
reduction~\cite{shamir84}. After that, many knapsack cryptosystems were
broken with low-density
attacks~\cite{low-density-subset-sum,low-density-subset-sum2}. 

More recently, \citet{impagliazzo} introduced a cryptographic scheme that is
provably as secure as \Ss. They proposed a function $f(\overrightarrow{a},S) =
\overrightarrow{a}, \sum_{i \in S} a_i \Mod{2^{l(n)}}$, i.e., the function which
concatenates $\overrightarrow{a}$ with the sum of the $a_i$'s for $i \in S$.
Function $f$ is a mapping of an $n$ bit string $S$ to an $l(n)$ bit string and
$\overrightarrow{a}$ are a fixed parameter. Our algorithms can be thought of as
an attempt to find a collision of such a function in the worst case.

However, in the average case more efficient algorithms are known. \citet{wagner02}
showed that when solving problems involving sums of elements from lists, one can
obtain faster algorithms when there are many possible solutions. In the
breakthrough paper, \citet{generic-knapsack1} gave $\Os(2^{0.337n})$ algorithm for an average case
\Ss. It was subsequently improved by \citet{generic-knapsack2} who gave an algorithm
running in $\Os(2^{0.291n})$. These papers introduced a \emph{representation}
technique that is a crucial ingredient in our proofs.

\karolw{TODO: write about LLL connection. JN: Maybe just skip?}

\paragraph{Total search problems:}

The \emph{Number Balancing} problem is: given $n$ real numbers $a_1,\ldots,a_n
\in [0,1]$, find two disjoint subsets $I,J \subseteq [n]$, such that the
difference $|\sum_{i \in I} a_i - \sum_{j \in J} a_j |$ is minimized. The
pigeonhole principle and the Chebyshev's inequality guarantee that there exists a solution with difference
at most $\Oh(\frac{\sqrt{n}}{2^n})$\todo{MM:why square root? KW: done}. \citet{karmarkar82} showed that in
polynomial time one can produce a solution with difference at most $n^{-
\Theta(\log{n})}$, but since then no further improvement is known.

\citet{ppad} considered the problem \emph{Equal Sums}: given $n$ positive
integers such that their total sum is less than $2^{n} - 1$, find two subsets
with the same sum. By the pigeonhole principle the solution always exists, hence
the decision version of this problem is obviously in P. However the hard part is
to actually find a solution. Equal Sums is in class PPP but it remains open to
show that it is PPP-complete. Recently, this question gained some momentum.
\citet{hoberg-ppp} showed that Number Balancing is as hard as \emph{Minkowski}.
\citet{ppp-reductions-2019} showed the reduction from Equal Sums to Minkowski
and conjectured that Minkowski is complete for the class PPP. Very recently,
\citet{sotiraki-ppp} identified the first natural problem complete for PPP.

In Appendix~\ref{number-balancing} we show that our techniques can also be used
to solve Number Balancing for integers in $\Os(1.7088^n)$ time.

\paragraph{Combinatorial Number Theory:}

If $\Sigma(S) < 2^{n}-1$, then by the pigeonhole principle the answer to the
decision version of \Ess on $S$ is always YES. In 1931 Paul Erd\H{o}s  was
interested in the smallest maximum value of $S$, such that the answer to \Ess on
$S$ is NO, i.e., he considered the function:

\begin{displaymath}
    f(n) = \min \{
        \max\{S\} \; | \; \text{all subsets of} \; S \; \text{are
        distinct}, \;\, |S| = n, \; \, S \subseteq \mathbb{N}
    \}
\end{displaymath}

and showed $f(n) > 2^n/(10\sqrt{n})$~\cite{erdos1941problems}.  The first nontrivial upper
bound on $f$ was $f(n) \le 2^{n-2}$ (for sufficiently large $n$)~\cite{conway1968sets}. Subsequently,
\citet{lunnon1988integer} proved that $f(n) \le 0.2246 \cdot 2^n$ and
\citet{combinatorial-bohman} showed $f(n) \le 0.22002 \cdot 2^n$.
\citet{erdHos1980survey} offered 500 dollars for proof or disproof of conjecture
that $f(n) \ge c 2^n$ for some constant $c$.

\paragraph{Other Variants:} \Ess has some connections to the study of the
structure of DNA
molecules~\cite{bio-connection,cieliebak2008complexity,cieliebak2003algorithms}.
\citet{cieliebak2003composing} considered $k$-\Ess, in which we need to find $k$
disjoint subsets of a given set with the same sum.  They obtained several
algorithms that depend on certain restrictions of the sets (e.g., small
cardinality of a solution). In the following work, \citet{cieliebak2008complexity}
considered other variants of \Ess and proved their NP-hardness.
 
\section{Preliminaries}

Throughout the paper we use the $\Os$ notation to hide factors polynomial in the input size and the $\Ot$
notation to hide factors logarithmic in the input size. We also use $[n]$ to denote the set $\{1,\ldots,n\}$.
If $S=\{a_1,\ldots,a_n\}$ is a set of integers and $X \subseteq
\{1,\ldots,n\}$, then $\Sigma_S(X):=\sum_{i \in X}a_i$. Also, we use $\Sigma(S) =
\sum_{s \in S} s$ to denote the sum of the elements of the set. We use the binomial coefficient notation for sets,
i.e., for a set $S$ the symbol ${S \choose k} = \{ X \subseteq S \; | \; |S| = k \}$ is the set of all subsets
of the set $S$ of size exactly $k$.

We may assume that the input to \Ess has the following properties:

\begin{itemize}
    \item the input set $S = \{a_1, \ldots, a_n \}$ consists of positive integers,
    \item $\sum_{i=1}^n a_i < 2^{\tau n}$ for a constant $\tau < 10$,
    \item integer $n$ is a multiple of $12$.
\end{itemize}

These are standard assumptions for \Ss (e.g., \cite{stacs2015, compression}). For
completeness, in Appendix~\ref{randomized-compression} we prove how to apply 
reductions to \Ess to ensure these properties.

We need the following theorem concerning the density of prime numbers~\cite[p. 371, Eq. (22.19.3)]{primes-bound}.
\todo{MM: my impression from looking at (22.19.3) is that it is asymptotically $2^b/b+o(sth)$, this
$o$ could be positive or negative, so not clear to me why the bound below is true.}

\begin{lemma}
    \label{lem:random-prime}
    For a large enough integer $b$, there exist at
    least $2^b/b$ prime numbers in the interval $[2^b,2^{b+1}]$.
\end{lemma}

The binary entropy function is $h(\alpha) = -\alpha \log_2{\alpha} -
(1-\alpha)\log_2{(1-\alpha)}$ for $\alpha \in (0,1)$ and $h(0) = h(1) = 0$. For
all integers $n \ge 1$ and $\alpha \in [0,1]$ such that $\sigma n$ is an
integer, we have the following upper bound on the binomial
coefficient~\cite{stirling-entropy}: ${{n}\choose{\alpha n}} \le 2^{h(\alpha)
n}$. We also need a standard bound on binary entropy function $h(x) \le
2\sqrt{x(1-x)}$.

Throughout this paper all logarithms are base 2.
 \section{Faster Exponential Space Algorithm}

In this section, we improve upon the meet-in-the-middle algorithm for \Ess.

\expspacethm

Theorem~\ref{expspace-alg} is proved by using two different algorithms for \Ess.
To bound the trade-off between these algorithms we introduce the
concept of a \emph{minimum solution}. 
\todo{JN: change `the minimal solution' to `a' minimal solution as one instance can have many. Maybe we even want to use `minimum' instead of `minimal' (but this is minor).}
\begin{definition}[Minimum Solution]
    For a set $S$ of positive integers we say that a solution $A,B \subseteq S$ is 
    a \emph{minimum} solution if its \emph{size} $|A|+|B|$ is smallest possible.
\end{definition}

We now assume that the size of the minimum solution has even size for
simplicity of presentation. The algorithm and analysis for the case of odd-sized 
minimum solution is similar, but somewhat more messy due to all the floors and
ceilings one needs to take care of.

In Section~\ref{unbalanced} we prove that the meet-in-the-middle approach
for \Ess already gives algorithm running in time $\Os((3-\eps)^{n/2})$ if the
minimum solution $A,B$ is \emph{unbalanced}, i.e., $||A\cup B| - \frac{2n}{3}| >
\eps' n$ for some $\eps' > 0$ depending on $\eps$. Subsequently, in Section~\ref{balanced} we propose an algorithm for \emph{balanced} instances, i.e., when the size of a
minimum solution is close to $2/3$.
\todo{MM: Are they really of roughly the same size, or is it only that their sum is $2/3n$?}
In particular, we show how to detect sets $A,B$ with $\Sigma(A)=\Sigma(B)$ and $|A|\approx |B|\approx\frac{n}{3}$, with an $\Os(2^{\frac{2}{3}n})$ time algorithm.
By bounding trade-off between the algorithms from Section~\ref{unbalanced} and Section~\ref{balanced} we prove Theorem~\ref{expspace-alg} and bound the running time numerically.

\subsection{\Ess for unbalanced solutions via meet-in-the-middle}
\label{unbalanced}

\todo{MM:What if $\ell$ is odd? Even the statement of the theorem does not make
sense, KW: I have added some sentences in the first paragraph. We can discuss if this is enough}
\begin{theorem}
    \label{unbalanced-alg}
    If $S$ is a set of $n$ integers with a minimum solution of size $\ell$, then
    $\Ess$ with input $S$ can be solved in $\Os({{n/2} \choose {\ell/2}} 2^{\ell/2})$ time with high probability.
\end{theorem}

\begin{algorithm}
    \caption{$\textsc{UnbalancedEqualSubsetSum}(S,\ell)$}
	\label{alg:unbalanced}
\begin{algorithmic}[1]
    \State Randomly split $S$ into two disjoint $S_1,S_2\subseteq S$, such that $|S_1| = |S_2| = n/2$
    \State Enumerate $C_1 = \{ \Sigma(A_1) - \Sigma(B_1) \; | \; A_1,B_1 \subseteq S_1, \; A_1 \cap B_1 = \emptyset , \; |A_1|+|B_1| = \ell/2 \}$
    \State Enumerate $C_2 = \{ \Sigma(A_2) - \Sigma(B_2) \; | \; A_2,B_2 \subseteq S_2, \; A_2 \cap B_2 = \emptyset , \; |A_2|+|B_2| = \ell/2 \}$
    \If {$\exists x_1 \in C_1, x_2 \in C_2$ such that $x_1 + x_2 = 0$}
        \State Let $A_1,B_1 \subseteq S_1$ be such that $x_1 = \Sigma(A_1) -
        \Sigma(B_1)$
        \State Let $A_2,B_2 \subseteq S_2$ be such that $x_2 = \Sigma(A_2) -
        \Sigma(B_2)$
        \State \Return $(A_1 \cup A_2, B_1 \cup B_2)$
    \EndIf
    \State \Return NO
\end{algorithmic}
\end{algorithm}

\begin{proof}[Proof of Theorem~\ref{unbalanced-alg}]
    Algorithm~\ref{alg:unbalanced} uses the meet-in-the-middle approach restricted
    to solutions of size $\ell$. We will show that this algorithm solves $\Ess$ 
    in the claimed running time. 

    The algorithm starts by randomly partitioning the set $S$ into two equally 
    sized sets $S_1,S_2$. Let $A,B$ be a fixed minimum solution of size 
    $|A\cup B|=\ell$. We will later show that \todo{MM:not sure about using whp
    here, KW: done?}
    with $\Omega(1/\poly(n))$ probability $|(A\cup B)\cap S_1| = |(A\cup B) \cap S_2| = \ell/2$. 
    We assume this is indeed the case and proceed with meet-in-the-middle.
    For $S_1$ we will list all $A_1,B_1$ that could possibly be equal to $S_1 \cap A$ and
    $S_1 \cap B$, i.e.\ disjoint and with total size $\ell/2$. We compute $x =
    \Sigma(A_1) - \Sigma(B_1)$ and store all these in $C_1$. We proceed analogously for $S_2$.

    We then look for $x_1 \in C_1$ and $x_2 \in C_2$ such that $x_1 + x_2 =
    0$. If we find it then we identify the sets $A_1$ and $B_1$ that correspond
    to $x_1$ and sets $A_2$ and $B_2$ that correspond to $x_2$ (the easiest way to do that
    is to store with each element of $C_1$ and $C_2$ the corresponding pair of sets when
    generating them). Finally we return $(A_1 \cup A_2, B_1 \cup B_2)$.

    \paragraph{Probability of a good split:} We now lower-bound the probability of
    $S_1$ and $S_2$ splitting $A\cup B$ in half.
    There are ${{n}\choose{n/2}}$ possible equally sized
    partitions. Among these there are ${{\ell}\choose{\ell/2}}{{n-\ell}\choose{(n-\ell)/2}}$ 
    partitions that split $A\cup B$ in half.
    The probability that a random partition splits $A$ and $B$
    in half is:

    \begin{displaymath}
        \frac{ {{\ell} \choose {\ell/2}} {{n-\ell} \choose {(n-\ell)/2}}
        }{{{n} \choose {n/2}}}
        \ge
        \frac{2^{\ell} 2^{n-\ell} }{(n+1)^2 2^{n}} = \frac{1}{(n+1)^2}
    \end{displaymath}
    because $\frac{2^{n}}{n+1} \le {{n}\choose{n/2}} \le 2^{n}$. 

    \paragraph{Running time:} To enumerate $C_1$ and $C_2$ we need
    $\Os({{n/2}\choose{\ell/2}} 2^{\ell/2})$ time, because first we guess set $S_1
    \cap (A\cup B)$ of size $\ell/2$ and then split between $A$ and $B$ in at most
    $2^{\ell/2}$ ways. We then check the existence
    of $x_1 \in C_1$ and $x_2 \in C_2$ such that $x_1 + x_2 = 0$ in 
    $\Os((|C_1| + |C_2|)\log{(|C_1| + |C_2|)})$ time by sorting. 

    We can amplify the probability of a \emph{good split} to $\Oh(1)$ by repeating the whole algorithm
    polynomially many times. \todo{MM:Talking about amplification here is not good, cause we cannot verify
    that split is good, so this is really not a whp split. We repeat the whole
algo. KW: I changed an order of the paragraphs. Is it ok now?}
    
    \paragraph{Correctness:} With probability $\Omega(1/\poly(n))$ we divide the $A\cup B$ equally 
    between $S_1$ and $S_2$. If that happens the set $C_1$ contains
    $x_1$ such that $x_1 = \Sigma(A\cap S_1) - \Sigma(B \cap S_1)$ and the set $C_2$
    contains $x_2$ that $x_2 = \Sigma(A \cap S_2) - \Sigma(B \cap S_2)$. Note that
    $x_1 + x_2 = \Sigma(A \cap S_1) + \Sigma(A \cap S_2) - \Sigma(B \cap S_1) -
    \Sigma(B \cap S_2) = \Sigma(A) - \Sigma(B)$ which is $0$, since $A,B$ is a solution.
    Therefore Algorithm~\ref{alg:unbalanced} finds a solution of size $\ell$ 
    (but of course, it could be different from $A$,$B$).
\end{proof}

\subsection{\Ess for balanced solutions}
\label{balanced}

\begin{theorem}
    \label{balanced-alg}
    Given a set $S$ of $n$ integers with a minimum solution size $\ell \in
    (\frac{1}{2}n,(1-\eps)n]$ for some constant $\eps > 0$, \Ess can be solved in time $\Os(2^{\ell})$ w.h.p.
\end{theorem}

We use Algorithm~\ref{alg:balanced} to prove 
Theorem~\ref{balanced-alg}. In this algorithm, we first pick a random prime $p$ in the range
$[2^{n-\ell},2^{n-\ell+1}]$, as well as an integer $t$ chosen
uniformly at random from $[1,2^{n-\ell}]$.  
We then compute the set $C = \{ X \subseteq S \; | \; \Sigma(X) \equiv_p t \}$. 
In the analysis, we argue that with $\Omega(1/\poly(n))$ probability $C$ contains two different
subsets $X,Y$ of $S$ with $\Sigma(X)=\Sigma(Y)$. To identify such pair it is 
enough to sort the set $|C|$ in time $\Oh(|C|\log{|C|})$, and then scan it.
We return $X \setminus Y$ and $Y \setminus X$ to guarantee that the returned sets
are disjoint.

\begin{algorithm}
    \caption{$\textsc{BalancedEqualSubsetSum}(a_1,\ldots,a_n,\ell)$}
	\label{alg:balanced}
\begin{algorithmic}[1]
    \State Pick a random prime $p$ in $[2^{n-\ell}, 2^{n-\ell+1}]$
    \State Pick a random number $t$ in $[1,2^{n-\ell}]$
    \State Let $C = \{ X \subseteq S \; | \; \Sigma(X) \equiv_p t \}$ be the set of candidates
    \Comment $C$ contains two sets with equal sum with probability $\Omega(1/\poly(n))$.
    \State Enumerate and store all elements of $C$
    \Comment In time $\Os(|C| + 2^{n/2})$
    \State Find $X,Y \in C$, such that $\Sigma(X) = \Sigma(Y)$
    \Comment In time $\Os(|C|)$
    \State \Return $(X\setminus Y, Y\setminus X)$
\end{algorithmic}
\end{algorithm}

We now analyse the correctness of Algorithm~\ref{alg:balanced}. Later, we will
give a bound on the running time and conclude the proof Theorem~\ref{balanced-alg}. 
First, observe the following:

\begin{lemma}
    \label{obs:sizeofcandidates}
    Let $S$ be a set of $n$ positive integers with minimum solution size of $\ell$.
    Let

    \begin{equation}\label{eq:sols}
        \Psi = \left\{\Sigma(X) \; | \; X \subseteq S \; \text{and} \;\;\exists Y \subseteq S \; \text{such
        that} \; X\neq Y \; \text{and} \;\; \Sigma(X)=\Sigma(Y)\right\}
        .
    \end{equation}
    If $\ell > \frac{n}{2}$, then $|\Psi| \geq 2^{n-\ell}$ (note that all
    elements in $\Psi$ are different).
\end{lemma}

\begin{figure}[ht!]
    \centering
   \begin{tikzpicture}[node distance=1em and 1em]
       \pgfdeclarepatternformonly{north east lines wide}
          {\pgfqpoint{-1pt}{-1pt}}
             {\pgfqpoint{10pt}{10pt}}
                {\pgfqpoint{9pt}{9pt}}
                   {
                            \pgfsetlinewidth{2pt}
                                 \pgfpathmoveto{\pgfqpoint{0pt}{0pt}}
                                      \pgfpathlineto{\pgfqpoint{9.1pt}{9.1pt}}
                                           \pgfusepath{stroke}
                                               }
       \def\widthrectangle{10.0}
       \def\heightrectangle{1.0}
       \def\sizea{3.0}
       \def\sizeb{7.0}
       \def\colorx{green!80}
       \def\sizex{1}

       \tikzstyle{chunk style}=[rounded corners=3pt]
       \tikzstyle{lines}=[pattern=north east lines wide]

       \newcommand{\highlight}[6]{
           \draw[decoration={brace,raise=3pt},decorate, thick] (#1,\heightrectangle) --
           node[above=5pt] {$#3$} (#2,\heightrectangle);
           \begin{scope}
               \clip[chunk style] (#1-#4,0) rectangle (#2+#5,\heightrectangle);
               \fill[pattern color=#6, lines]   (#1,0)   rectangle (#2,\heightrectangle);
           \end{scope}
       }
       \newcommand{\drawx}[3]{
           \fill[color=#3, chunk style]   (#1,0)     rectangle (#1+#2,\heightrectangle);
       }

       \highlight{0}{\sizea}{A}{1}{0}{red}
       \highlight{\sizeb}{\widthrectangle}{B}{0}{1}{blue}

       \node (end_rectangle) at (\widthrectangle,\heightrectangle) {};
       \node (end_a)         at (\sizea,\heightrectangle) {};
       \node (begin_b)       at (\sizeb,0) {};

       \drawx{4.5}{0.2}{\colorx}
       \drawx{5.15}{0.2}{\colorx}
       \drawx{5.8}{0.2}{\colorx}

       \draw[decoration={brace,raise=3pt},decorate,thick] (4.5,\heightrectangle) -- node[above=5pt] {$X$} (6.0,\heightrectangle);

       \draw[thick] (0,0) rectangle (end_rectangle);

       \draw[decoration={mirror,brace,raise=5pt},decorate,thick] (0,0) -- node[below=8pt] {$S$} (\widthrectangle,0);

   \end{tikzpicture}
    \caption{Scheme presents the set $S$ of positive integers and two disjoint
    subsets $A,B \subseteq S$. The point is that if $\Sigma(A) = \Sigma(B)$
    then for any subset $X \subseteq S \setminus (A \cup B)$ we have a guarantee
    that $\Sigma(A\cup X) = \Sigma(B \cup X)$.}
   \label{fig:combinatoria-lamma}
\end{figure}
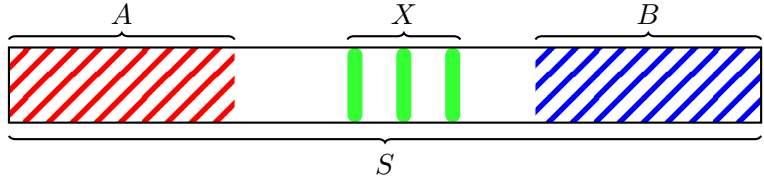

\begin{proof}
    Let $A,B \subseteq S$ be a fixed minimum solution to $S$. We know that $\ell
    = |A \cup B|$, $\Sigma(A) = \Sigma(B)$ and $A \cap B = \emptyset$.
    With this in hand we construct set $\Psi$ of $2^{n-\ell}$ pairs of different $X,Y
    \subseteq S$ with $\Sigma(X) = \Sigma(Y)$.

    Consider set $Z = S\setminus(A\cup B)$. By the bound on the size of $A$ and $B$ we know
    that $|Z| = n-\ell$. Now we construct our candidate pairs as follows: take
    any subset $Z' \subseteq Z$ and note that $X \cup Z'$ and $Y \cup Z'$ satisfy 
    $\Sigma(X \cup Z') = \Sigma(Y\cup Z')$. There are $2^{|Z|}$ possible subsets 
    of set $Z$ and the claim follows.
    
    Now we will prove that if $\ell > \frac{n}{2}$ then all subsets of $Z$ have
    a different sum.  Assume for a contradiction that there exist $Z_1,Z_2
    \subseteq Z$, such that $\Sigma(Z_1) = \Sigma(Z_2)$ and $Z_1 \neq Z_2$.
    Then $Z_1\setminus Z_2$ and $Z_2\setminus Z_1$ would give a solution smaller
    than $A,B$, because $|Z| < \ell$. This contradicts the assumption about the
    minimality of $A,B$. It follows that if $\ell>\frac{1}{2}n$ then all
    constructed pairs have a different sum.

\end{proof}

\todo{MM: We use the term solutions here, this is wrong. Karol, could you 
propose a fix here?}
Now, we consider the hashing function $h_{t,p}(x) = x + t \Mod{p}$. We prove
that if the set $\Psi$ (see Equation~\ref{eq:sols}) is sufficiently large, 
then for a random choice of $t$, at least one element of set $\Psi$ is in 
the congruence class $t$.

\todo{JN: I added the proof of this lemma below, but I guess we could also select a prime from the interval $[2^{n-\ell},2^{n-\ell+1}]$, KW: done}
\begin{lemma}
    \label{claim:prob}
    Let $S$ be the set of $n$ positive integers bounded by $2^{\Oh(n)}$ with
    minimum solution of size $\ell$ and $\ell > \frac{n}{2}$. For a random prime $p \in
    [2^{n-\ell},2^{n-\ell+1}]$ and a random $t \in [1,2^{n-\ell}]$ let $C_{t,p} = \{ X \subseteq
    S \; | \; \Sigma(X) \equiv_p t \;\}$. Then,

    \begin{displaymath}
        \prob{ \exists X,Y \in C_{t,p} \; \Big| \; \Sigma(X) = \Sigma(Y), \; X \neq Y }{t,p}
        \ge \Omega(1/n^2).
    \end{displaymath}
\end{lemma}
\todo{karolw: I changed notations from of sets A,C in the proof because of notation conflict, with set $C_{t,p}$ and solution set.}
\begin{proof}
	Let $\Psi$ be the set defined in~\eqref{eq:sols}.
	So $\Psi \subseteq \{1,\ldots,2^{\Oh(n)}\}$, and $|\Psi|\geq 2^{n-\ell}$.
    \todo{MM: I don't understand this remark about constant probability. Isn't it enough to show this with $1/n^3$ probability as in the claim?}
	It is sufficient to bound the probability, that there exists an element $a \in \Psi$ such $a\equiv_p t$.
	Let $a_1,a_2 \in \Psi$ be two distinct elements.
	\[
		\prob{a_1 \equiv_p a_2}{p}=\prob{p \text{ divides }|a_1 - a_2|}{p} \leq \Oh(n (n-\ell) / 2^{n-\ell} ).
	\]
    This is because $|a_1-a_2|$ can only have $\Oh(n)$ prime divisors, and we
    are sampling $p$ from the set of at least
    $2^{n-\ell}/(n-\ell)$ primes by Lemma~\ref{lem:random-prime}.
	Let $k$ be the number of pairs $a_1,a_2 \in \Psi$ such that $a_1 \equiv_p
    a_2$. We have $\Ex{k} \leq \Oh( |\Psi|+(|\Psi|n)^2/2^{n-\ell})$.
    We know that $|\Psi| \ge 2^{n-\ell}$, so $\frac{|\Psi|^2}{2^{n-\ell}} \geq
    |\Psi|$ which means that $\Ex{k} \le \Oh((|\Psi|n)^2/2^{n-\ell})$.
    Hence, by Markov's inequality $k$ is at most $\Oh((|\Psi|n)^2/2^{n-\ell})$ with at least constant probability.
    If this does indeed happen, then
    
    \begin{displaymath}
        |\{a \Mod{p} \; | \;  a \in \Psi\}| \geq \frac{|\Psi|^2}{k} 
        \geq \Omega\left(\frac{|\Psi|^2}{(|\Psi| n)^2 / 2^{n-\ell}}\right)
        \geq \Omega(2^{n-\ell}/n^2)
        ,
    \end{displaymath}
    and the probability that $t$ chosen uniformly at random from
    $[1,2^{n-\ell}]$ will be among one of the elements of set $\{a \Mod{p} \; | \; a \in \Psi\}$ is $|\{a \Mod{p}
    \; | \; a \in \Psi\}|/2^{n-\ell} \geq \Omega(1/n^2)$.
		
\end{proof}

\begin{proof}[Proof of correctness of Algorithm~\ref{alg:balanced}]
    By Lemma~\ref{claim:prob}, after choosing a random prime $p$ and random
    number $t \in [1,2^{n-\ell}]$ the set $ C = \{ X \subseteq S \; | \; \Sigma(X) \equiv_p t \}$ contains at least
    two subsets $X,Y \subseteq S$, such that $\Sigma(X) = \Sigma(Y)$ with 
    probability $\Omega(1/\poly(n))$. \todo{MM:whp again, KW: Done?} Algorithm~\ref{alg:balanced} computes the set
    $C$ and finds $X',Y' \subseteq S$, such that $\Sigma(X') = \Sigma(Y')$. Then it returns the solution 
    $X'\setminus Y',Y'\setminus X'$.
\end{proof}

Now we focus on bounding the running time of
Algorithm~\ref{alg:balanced}. We start by bounding the size of the candidate set $C$.

\todo{MM:Doesn't this break if there are actually many more solutions? There could be.}
\begin{claim}
    \label{claim:sizeofl}
    Let $S$ be the set of $n$ non-negative integers bounded by $2^{\Oh(n)}$ with
    a minimum solution of size $\ell$ such that $\ell \le (1-\eps)n$ for some
    constant $\eps>0$ (think of $\eps = 1/100$). For a random prime $p \in
    [2^{n-\ell},2^{n-\ell+1}]$ and a random number $t \in [1,2^{n-\ell}]$ let $C_{t,p} = \{ X \subseteq
    S \; | \; \Sigma(X) \equiv_p t \;\}$. Then

    \begin{displaymath}
        \Ex{|C_{t,p}|} \le \Os(2^{\ell})
    \end{displaymath}
\end{claim}

\todo{Two minor issues here: (1) missing n in the bound, add it or change to O*
(2) the proof conditions on t, but then p is not random, to avoid this maybe choose
t in $[0,2^{n-l}]$ and not $[0,p]$?}
\begin{proof}
    By the linearity of expectations:
    \begin{displaymath}
        \label{proof-claim-sizeof}
        \Ex{|C_{t,p}|} = \sum_{X \subseteq S} \prob{p \; \text{divides} \; \Sigma(X) - t}{t,p}
    \end{displaymath}

    For the remaining part of the proof we focus on showing $\prob{p \; \text{divides} \;
    \Sigma(X) - t}{t,p} \le \Os(2^{\ell - n})$ for a fixed $X \subseteq S$. It
    automatically finishes the proof, because there are $2^n$ possible subsets $X$.

    We split the terms into two cases. If $\Sigma(X) = t$, then
    $p$ divides $\Sigma(X) - t$ with probability $1$. However, for a fixed $X \subseteq
    S$, the probability that $\Sigma(X) = t$ is $\Oh(\frac{1}{2^{n-\ell}})$
    because $t$ is a random number from $[1,2^{n-\ell}]$ and $p \ge 2^{n-\ell}$.

    On the other hand, if $\Sigma(X) \neq t$, then by the assumption, the set $S$
    consists of non-negative integers bounded by $2^{\tau n}$ for some
    constant $\tau > 0$. In particular, $|\Sigma(X) - t| \le 2^{\tau n}$. This
    means that $|\Sigma(X) - t|$ has at most $\frac{\tau n}{n-\ell} \le
    \frac{\tau}{\eps} = \Oh(1)$ prime factors of size at least $2^{n-\ell}$. Any
    prime number $p$ that divides $\Sigma(X) - t$ must therefore be one of
    these numbers. By Lemma~\ref{lem:random-prime} there are at least
    $2^{n-\ell}/(n-\ell)$ prime numbers in range $[2^{n-\ell},2^{n-\ell+1}]$. Hence, for a
    fixed $X \subseteq S$ the probability that $p$ divides $\Sigma(X) - t$ is bounded by
    $\Oh(n2^{\ell-n})$.
\end{proof}

\begin{lemma}
    \label{modenumeration}
    The set $C_{t,p}$ can be enumerated in time $\Os\left(\max \left\{|C_{t,p}|,
    2^{n/2}\right\}\right)$.
\end{lemma}

The proof of the above lemma is based on \citet{schroeppel} algorithm for \Ss.
For a full proof of Lemma~\ref{modenumeration} see, e.g., Section 3.2
of~\cite{generic-knapsack2}. Observe, that for our purposes the
running time is dominated by $\Os(|C_{t,p}|)$.

\todo{karolw: I used algorithm $\Os(|C| + 2^{n/2})$ only because I found a clean writeup in
the literature (and I did not find $\Os(|C| + p)$)}

\begin{proof}[Proof of the running time of Algorithm~\ref{alg:balanced}]
    To enumerate the set $C_{t,p}$ we need $\Os(|C| + 2^{n/2})$ time (see
    Lemma~\ref{modenumeration}). To find two subsets $X,Y \in C$, such that
    $\Sigma(X) = \Sigma(Y)$ we need $\Os(|C|\log{|C|})$ time: we sort 
    $C$ and scan it.

    The prime number $p$ is at most $2^{n-\ell+1}$ and the expected size
    of $C$ is $\Os(2^{\ell})$. Because we assumed that $\ell > \frac{n}{2}$
    the expected running time is $\Os(2^\ell)$ (we can terminate algorithm when it
    exceeds $\Os(2^\ell)$ to Monte Carlo guarantees). The probability of success is
    $\Omega(1/\poly(n))$. We can amplify it with polynomial overhead to any constant by
    repetition.

\end{proof}

This concludes the proof of Theorem~\ref{balanced-alg}.

\subsection{Trade-off for \Ess}

In this section, we will proof the Theorem~\ref{expspace-alg} by combining
Theorem~\ref{balanced-alg} and Theorem~\ref{unbalanced-alg}.

\begin{proof}[Proof of Theorem~\ref{expspace-alg}]
    Both Theorem~\ref{balanced-alg} and Theorem~\ref{unbalanced-alg} solve \Ess.
    Hence, we can focus on bounding the running time. By the trade-off between Theorem~\ref{balanced-alg} (which works
    for $\ell \in (\frac{n}{2},(1-\eps)n)$ and Theorem~\ref{unbalanced-alg} the running time is:
    \todo{MM: This is a bit unintuitive, perhaps better to first write sums, and then bound them by maxs?}

    \begin{displaymath}
        \Os\left(
        \max_{\ell \in [1,n/2] \cup [(1-\eps)n,n]} \left\{
            {{n/2} \choose {\ell/2}} 2^{\ell/2}
            \right\}
        + \max_{\ell \in (n/2,(1-\eps)n)} \left\{ \min \left\{ 
                {{n/2} \choose {\ell/2}} 2^{\ell/2},
                2^\ell
        \right\}
        \right\} 
        \right)
    \end{displaymath}
    
    For simplicity of analysis we bounded the sums by the maximum (note that
    $\Os$ notation hides polynomial factors). \karolw{KW: maybe this sentence is enough?}
    When $\ell \le n/2$, the running time is maximized for
    $\ell = n/2$, because (let $\ell = \alpha n$):

    \begin{displaymath}
        \Os \left({{n/2} \choose {\ell/2}} 2^{\ell/2} \right) = \Os \left( 2^{\frac{n}{2} ( h(\alpha) +
        \alpha)} \right)
    \end{displaymath}
    and the entropy function $h(x)$ is increasing in range $[0,0.5)$. For $\ell =
    \frac{n}{2}$ the running time is $\Os(2^{0.75n}) \le \Os(1.682^n)$.
    Similarly, we get a running time superior to the claimed one when $\ell \in [(1-\eps)n, n]$.
    \todo{MM:This is a standard inequality but maybe it should go into prelims?.  KW:I copied it to preliminaries.}
    Note that $h(x) \le 2\sqrt{x(1-x)}$,  which means that the running time is bounded by
    $\Os(2^{\frac{n}{2}(h(1-\eps) + (1-\eps))}) \le
    \Os(2^{\frac{n}{2}(1+2\sqrt{\eps})})$ which is
    smaller than our running time for a sufficiently small constant $\eps$.
    
    Finally, when $\ell \in [n/2,(1-\eps)n]$ we upper bound the running time by the:

    \begin{displaymath}
        \Os \left( \max_{\ell \in [n/2,(1-\eps)n)} \left\{ \min \left\{ 
            2^{\frac{n}{2} ( h(\alpha) + \alpha)}, 2^{\alpha n}
        \right\}
        \right\}
        \right)
        .
    \end{displaymath}
    
    The above expression is maximized when $h(\alpha) = \alpha$. By numeric calculations $\alpha <
    0.77291$, which gives the final running time $\Os(2^{\alpha n }) \le
    \Os(1.7088^n)$.
\end{proof}
 \section{Polynomial Space Algorithm}
\label{polyspacealg:sec}

The naive algorithm for \Ess in polynomial space works in $\Os(3^n)$ time. We
are given a set $S$. We guess a set $A \subseteq S$ and then guess a set $B \subseteq
S\setminus A$. Finally, we check if $\Sigma(A) = \Sigma(B)$. The running time
is:

\begin{displaymath}
    \Os\left( {{|S|}\choose{|A|}} {{|S| - |A|}\choose{|B|}}\right) \le \Os(3^n)
    .
\end{displaymath}

Known techniques for \Ss allow us to get an algorithm running in $\Os(2^{1.5n})$ and polynomial space.

\begin{theorem}
    \label{naive-polyspace}
    There exists a Monte Carlo algorithm which solves \Ess in polynomial space
    and $\Os(2^{1.5n}) \le \Os(2.8285^n)$ time. The algorithm assumes random read-only access to
    exponentially many random bits.
\end{theorem}

A crucial ingredient of Theorem~\ref{naive-polyspace} is a nontrivial result
for the \emph{Element Distinctness}
problem~\cite{polyspace-stoc2017,beame-focs2013}. In this problem, one is given
read-only access to the elements of a list $x \in [m]^n$ and the task is to find
two different elements of the same value. The problem can be naively solved in
$\Oh(n^2)$ time and $\Oh(1)$ space by brute force. Also by sorting, we can solve
Element Distinctness in $\Ot(n)$ time and $\Ot(n)$ space. \citet{beame-focs2013}
showed that the problem can be solved in $\Ot(n^{3/2})$ randomized time and $\Ot(1)$ space.
The algorithm assumes access to a random hash function $f : [m] \rightarrow
[n]$.

\begin{proof}[Proof of Theorem~\ref{naive-polyspace}]
    \todo{MM: I am still confused about this. We use ED as a black-box. What Beame's ED is
    doing is it is accessing the elements in a random fashion. There is nowhere in this algorithm
    any need to go to the next element of the list. Why are we doing it here?}

    We can guarantee random access to the list $L = 2^S$ of all subsets of the
    set $S = \{a_1,\ldots,a_n \}$ on the fly. Namely, for a pointer $x \in \{0,1\}^n$ we can return an element
    of the list $L$ that corresponds to $x$ in $\Os(1)$ time by choosing
    elements $a_i$ for which $x_i=1$. More precisely:

    \begin{displaymath}
        L(x_1,\ldots,x_n) = \{ a_i \; | \; i \in [n], \; x_i = 1\}
        .
    \end{displaymath}

    Now to decide \Ess on set $S$ we execute the Element Distinctness algorithm
    on the list $L$ of sums of subsets. The list has size $2^n$, hence the
    algorithm runs in $\Os(2^{1.5n})$ time. Element Distinctness uses only
    polylogarithmic space in the size of the input, hence our algorithm uses
    polynomial space.
\end{proof}

Quite unexpectedly we can still improve upon this algorithm.

\subsection{Improved Polynomial Space Algorithm}

In this section, we show an improved algorithm.

\polyspacethm

Similarly to the exponential space algorithm for \Ess, we will combine two
algorithms. We start with a generalization of Theorem~\ref{naive-polyspace}
parametrized by the size of the solution.
\todo{MM: What does "for a different sizes of solution" mean? KW: done, changed
to ``parametrized by''.}

\begin{lemma}
    \label{polyspace-unbalanced}
    Let $S$ be a set of $n$ positive integers, $A,B \subseteq S$ be the
    solution to \Ess (denote $a = |A|$ and $b = |B|$). There exists a Monte
    Carlo algorithm which solves \Ess in polynomial space and time

    \begin{displaymath}
        \Os\left(\left({{n}\choose{a}} + {{n} \choose{b}}\right)^{1.5}\right)
        .
    \end{displaymath}
    The algorithm assumes random read-only access to exponentially many random
    bits.
\end{lemma}

\begin{proof}
    The proof is just a repetition of the proof of Theorem~\ref{naive-polyspace}
    for a fixed sizes of solutions. Our list $L$ will consists of all subsets
    ${{S}\choose{a}}$ and ${{S}\choose{b}}$. Then we run Element Distinctness
    algorithm, find any sets $A,B \in L$ such that $\Sigma(A) = \Sigma(B)$ and
    return $A\setminus B, B\setminus A$ to make them disjoint.

    The running time follows because Element Distinctness runs in time
    $\Ot(n^{1.5})$ and $\polylog(n)$ space.
\end{proof}

Note that the runtime of Lemma~\ref{polyspace-unbalanced} is maximized when
$|A| = |B| = n/2$. The next algorithm gives improvement in that case.

\begin{lemma}
    \label{polyspace-balanced}
    Let $S$ be a set of $n$ positive integers, $A,B \subseteq S$ be the
    solution to \Ess (denote $a = |A|$ and $b = |B|$). There exists a Monte
    Carlo algorithm which solves \Ess in polynomial space and time

    \begin{displaymath}
        \Os\left( \min \left\{ 
            {{n}\choose{a}} 2^{0.75(n-a)},
            {{n}\choose{b}} 2^{0.75(n-b)}
        \right\}
        \right)
        .
    \end{displaymath}
    The algorithm assumes random read-only access to exponentially many random
    bits.
\end{lemma}

\begin{proof}[Proof of Lemma~\ref{polyspace-balanced}]
    Without loss of generality, we focus on the case $a \le b$. First we
    guess a solution set $A \subseteq S$. We answer YES if we find set $B
    \subseteq S\setminus A$ such that $\Sigma(A) = \Sigma(B)$ or find two disjoint
    subsets with equal sum in $S\setminus A$. We show that we can do it in $\Os(2^{0.75(|S\setminus A|)})$
    time and polynomial space which finishes the proof.

    First, we arbitrarily partition set $S\setminus A$ into two equally sized sets $S_1$ and
    $S_2$. Then we create a list $L_1 = [ \Sigma(X) \; | \; X \subseteq S_1]$
    and list $L_2 = [ \Sigma(A) - \Sigma(X) \; | \; X \subseteq S_2]$. We do not construct them explicitly because it would take
    exponential space. Instead we provide a read-only access to them (with the counter
    technique). We run Element Distinctness on concatenation of $L_1$ and $L_2$. 
    If element distinctness found $x \in L_1$ and $y \in L_2$ such that $x=y$, then we backtrack
    and look for $X \subseteq S_1$, such that $\Sigma(X) = x$ and $Y \subseteq
    S_2$, such that $\Sigma(Y) = \Sigma(A) - y$ and return $(A,X\cup Y)$ which is a
    good solution, because $\Sigma(Y) + \Sigma(X) = \Sigma(A)$.
    
    In the remaining case, i.e.\ when Element Distinctness finds a duplicate only in one of the lists
    then, we get a feasible solution as well. Namely, assume that Element Distinctness finds $x,y
    \in L_1$ such that $x=y$ (the case when $x,y \in L_2$ is analogous). Then we
    backtrack and look for two corresponding sets $X,Y \subseteq L_1$ such that
    $X \neq Y$ and $\Sigma(X) = \Sigma(Y) = x$. Finally we return $(X\setminus Y,Y\setminus X)$.

    For the running time, note that the size of the list $|L_1| = |L_2| =
    2^{0.5|S\setminus A|}$. Hence Element Distinctness runs in time $\Os((|L_1| +
    |L_2|)^{1.5}) = \Os(2^{0.75(n-a)})$. The backtracking takes time $\Os(|L_1|
    + |L_2|)$ and polynomial space because we scan through all subsets of $S_1$
    and all subsets of $S_2$ and look for a set with sum equal to the known value. 
\end{proof}

\begin{proof}[Proof of Theorem~\ref{improved-polyspace}]
    By trade-off between Lemma~\ref{polyspace-balanced} and
    Lemma~\ref{polyspace-unbalanced} we get the following running time:

    \begin{displaymath}
        \Os\left(\max_{1 \le a,b \le n} \left\{ \min \left\{
            \left({{n}\choose{a}} + {{n}\choose{b}}\right)^{1.5},
            {{n}\choose{a}} 2^{0.75(n-a)},
            {{n}\choose{b}} 2^{0.75(n-b)}
        \right\}
        \right\}
        \right)
    \end{displaymath}

    By symmetry this expression is maximized when $a = b$. Now we will write the
    exponents by using entropy function (let $a = \alpha n$):

    \begin{displaymath}
        \Os \left( \max_{\alpha \in [0,1]} \left\{ \min \left\{
            2^{1.5 h(\alpha)n}, 2^{(h(\alpha) + 0.75(1-\alpha))n}
            \right\}
            \right\}
            \right)
    \end{displaymath}
    
    The expression is maximized when $1.5 h(\alpha) = h(\alpha) +
    0.75(1-\alpha)$, By numerical computations $\alpha < 0.36751$, which
    means that the running time is $\Os(2^{1.42312n}) \le \Os(2.6817^n)$.
    
\end{proof}
 
\section{Conclusion and Open Problems}

In this paper, we break two natural barriers for \Ess: we propose an
improvement upon the meet-in-the-middle algorithm and upon the polynomial space
algorithm. Our techniques have additional applications in the problem of finding
collision of hash function in cryptography and the number balancing problem (see
Appendix~\ref{number-balancing}).

We believe that our algorithms can potentially be improved
with more involved techniques. However, getting close to the running time of \Ss
seems ambitious. In Appendix~\ref{conditional-lower-bound} we show that
a faster algorithm than $\Os(1.1893^n)$ for \Ess would yield a
faster than $\Os(2^{n/2})$ algorithm for \Ss. It is quite far from our bound
$\Os(1.7088^n)$. The main open problem is therefore to close the gap between
upper and lower bounds for \Ess.

\todo{MM:As I said earlier I find the following part confusing KW: I
reformulated this question.}

 \section{Acknowledgment}

The authors would like to thank anonymous reviewers for their remarks and 
suggestions. This research has been initiated during Parameterized Algorithms Retreat of
University of Warsaw 2019, Karpacz, 25.02-01.03.2019.

\bibliographystyle{abbrvnat}
\bibliography{bib}

\begin{appendices}

\section{Preprocessing and Randomized Compression for \Ess}
\label{randomized-compression}

We will repeat the arguments from~\cite{compression}. Similar arguments are also
present in~\cite{stacs2015,stacs2016}.

\begin{theorem}
    Given set $S$ of $n$ integers $a_1,\ldots,a_n \in \{-2^m,\ldots, 2^m\}$ (with $m \gg
    n$). In $\Oh(\text{poly}(n,m))$ time we can construct a set $S'$ that consists of $n$
    positive integers in $\{1,\ldots,2^{8n}\}$ such that:

    \begin{displaymath}
        \exists X,Y \subseteq S,\, \text{such that}\;\; \Sigma(X) =
        \Sigma(Y)\;\;
        \text{iff}\;\; \exists X',Y' \subseteq S'\, \text{such that}\;\; \Sigma(X')
        = \Sigma(Y')
    \end{displaymath}
    with probability at least $1-2^{-n}$ or we will solve \Ess on that instance
    in polynomial time.

\end{theorem}

\begin{proof}
    If $0 \in S$, then we immediately answer YES, because sets $A = \{0\}$ and $B
    = \emptyset$ are a proper solution to \Ess. If $m \ge 2^n$, then
    the algorithm running in time $\Oh(m4^n) \ge \Oh(m^3)$ runs in polynomial
    time of the instance size. Hence we can assume that $m < 2^n$ and $0 \notin
    S$.
    
    Pick a random prime $p \in [2^{7n}, 2^{8n}]$. We will transform our original
    instance $S$ into an instance $S'$ in the following way:

    \begin{displaymath}
        a'_i \equiv a_i \Mod{p}
    \end{displaymath}

    for all $i \in [n]$. In particular it means that all numbers in $S'$ are
    positive and smaller than $2^{8n}$. Observe, that if there is a solution for \Ess on
    instance $S$, then the same set of indices is also a solution for \Ess on
    instance $S'$. On the other hand, we want to show that if an answer to \Ess
    on original instance $S$ was NO, then for all pairs of subsets $A,B
    \subseteq S'$ it will hold that $\Sigma(A) \neq \Sigma(B)$.

    For some $I,J \subseteq [n]$, in order to get $\sum_{i \in I} a'_i = \sum_{j
    \in J} a'_j$, while $\sum_{i \in I} a_i \neq \sum_{j \in J} a_j$, it must be
    that $p$ is a divisor of $D(I,J) = \sum_{i \in I} a_i - \sum_{j \in J} a_j$.
    We will call such prime numbers \emph{bad}.

    There are $2^{2n}$ possible pairs of $I,J \subseteq [n]$. For a fixed $I,J
    \subseteq [n]$ there are at most $\log{(n2^m)}$ bad primes (because $D(I,J)
    \le n 2^m$). Hence there are at most:

    \begin{displaymath}
        2^{2n} \log{(n 2^{m})} \le 2^{2n} (m + \log{n}) \le 2^{4n}
    \end{displaymath}

    possible bad primes. By Lemma~\ref{lem:random-prime}, the prime
    number $p$ is taken from the range containing at least $2^{7n}$ primes.
    Therefore, for every $I,J \subseteq [n]$ it holds that:

    \begin{displaymath}
        \prob{\sum_{i \in I} a'_i = \sum_{j \in J} a'_j}{p} \le 2^{-3n}
        .
    \end{displaymath}

    By talking union bound over all possible $2^{2n}$ pairs of $I,J \subseteq
    [n]$ the probability of error is bounded by $2^{-n}$.
\end{proof}

What is left to prove, is that we can assume, that $n$ is divisible by $12$.  By
the above Lemma we know that $S$ consists of only positive numbers. Let $M$ be
$\Sigma(S) + 1$. Observe that we can always add numbers from set  $Z =
\{M,2M,4M,8M\ldots,\}$ and the answer to \Ess on the modified instance will not
change because numbers in $Z$ always have a different sum. Moreover, none of the
subset of $S$ can be used with numbers from $Z$, because $\Sigma(S) < M$. Hence we
can always guarantee that $n$ is divisible by 12 by adding appropriate amount of
numbers from $Z$.  Namely, note that if $n \equiv k \Mod{12}$, for some $k \neq
0$, then we can add $k$ numbers to the original instance $S$ and the answer to
the \Ess will not change.
 \section{Sharper Reduction from \Ss}
\label{conditional-lower-bound}

In this section, we show a direct reduction from \Ss. As far as we know, it is
slightly sharper than currently known reduction~\cite{woeginger92} (in terms of
constants in the exponent).

\begin{theorem}
    If $\Ess$ can be solved in time $\Os((2-\eps)^{0.25n})$ for some $\eps>0$ then \Ss can be solved
    in time $\Os((2-\eps')^{0.5n})$ for some constant $\eps'>0$.
\end{theorem}

\begin{proof}
    Assume that we have a black-box access to the algorithm for \Ess running in time
    $\Os((2-\eps)^{0.25n})$ for some $\eps > 0$. We will show how to use it to get an algorithm for
    \Ss running in time $\Os((2-\eps)^{0.5n})$.

    Given an instance $S,t$ of \Ss such that $S = \{a_1,\ldots,a_n\}$, we will construct an
    equivalent instance $S'$ of \Ess such that $S' = \{s_1,\ldots,s_{2n+1}\}$.
    Note, that for the running time this will be enough. The construction is as
    follows:

    \begin{itemize}
        \item for $1 \le i \le n$, let $s_i = a_i \cdot 10^{n+1} + 2 \cdot 10^i$,
        \item for $1 \le i \le n$, let $s_{i+n}  = 1 \cdot 10^i$,
        \item let $s_{2n+1} = t \cdot 10^{n+1} + \sum_{i=1}^n 1 \cdot 10^i$.
    \end{itemize}

    First let us prove that if $(S,t)$ is a YES instance of \Ss then $S'$ is a
    YES instance for \Ess. Namely let $X \subseteq [n]$, such that $\sum_{i \in
    X} a_i = t$.  Then, sets $A = \{ s_i \; | \; i \in X\} \cup \{s_{i+n} \; |
    \; i \notin X \}$ and $B = \{ s_{i+n} \; | \; i \in X \} \cup \{s_{2n+1}\}$
    are a good solution to \Ess on instance $S'$, because $\sum(A) = \sum(B)$
    and $A \cap B = \emptyset$.

    Now for other direction, we will prove that if $S'$ is a YES instance of
    \Ess then $(S,t)$ is a YES instance of \Ss. Assume that $S'$ is a YES
    instance and a pair $A,B \subseteq S'$ is a correct solution. Observe that
    if for some $i \le n$ element $s_i \in A$ then $s_{2n+1} \in B$. It is
    because the sets $A,B$ have an equal sum and only elements
    $s_i,s_{i+n},s_{2n+1}$ have something nonzero at the $i$-th decimal place.
    Moreover all smaller decimal places of all numbers sum up to something
    smaller than $10^i$ and therefore cannot interfere with the place $10^i$.

    Finally observe, that numbers $s_{i+n}$ for $i \in [n]$ cannot produce a YES
    instance on their own. Hence sets $A\cup B$ contain at least one number
    $s_i$ for $i \in [n]$. WLOG let $A$ be the set that contains such an $s_i$.
    Then set $B$ has to contain $s_{2n+1}$. It means that set $B$ cannot contain
    any $s_i$ for $i \in [n]$. 

    In particular $\sum(A) / 10^{n+1} = \sum(B) / 10^{n+1}$. Only numbers $s_i$
    for $i \in [n]$ contribute to $\sum(A)/10^{n+1}$ and only number $s_{2n+1}$
    contributes to the $\sum(B)/10^{n+1}$. Hence there exists a subset $Z
    \subseteq S$, such that $\sum(Z) = t$.
\end{proof}
 \section{Folklore \Ess by 4-SUM with better memory}
\label{table-4-sum-ess}

\begin{theorem}
    \Ess can be solved in deterministic $\Os(3^{n/2})$ time and $\Os(3^{n/4})$ space.
\end{theorem}

\begin{proof}
    First, we arbitrarily partition $S$ into $S_1 = \{a_1,\ldots,a_{n/4}\}$, $S_2 = \{a_{n/4+1},\ldots,a_{n/2}\}$ $S_3 =
    \{a_{n/2+1},\ldots,a_{3n/4}\}$ and $S_4 = \{a_{3n/4+1},\ldots,a_{n}\}$.
    Denote the vectors that correspond to these sets by
    $\overline{a}_1,\ldots,\overline{a}_4 \in \mathbb{Z}^{n/4}$, i.e.,

    \begin{displaymath}
        \overline{a}_i = (a_{(i-1)n/4+1},a_{(i-1)n/4+2},\ldots,a_{in/4}) \;\;
        \text{for} \; \; i \in \{1,2,3,4\}
        .
    \end{displaymath}
    
    Recall that in \Ess we were looking for two subsets $A,B \subseteq S$, such
    that $A \cap B = \emptyset$ and $\Sigma(A) = \Sigma(B)$. We can split the
    solution to 8 subsets:

    \begin{displaymath}
        A_i := S_i \cap A \;\; \text{and} \;\; B_i := S_i \cap B \;\; \text{for} \; i \in \{1,2,3,4\}
        .
    \end{displaymath}

    Then, the equation for the solution is:

    \begin{displaymath}
        \Sigma(A_1) + \Sigma(A_2) + \Sigma(A_3) + \Sigma(A_4) = \Sigma(B_1) + \Sigma(B_2) + \Sigma(B_3) + \Sigma(B_4).
    \end{displaymath}

    We can rewrite it as:

    \begin{displaymath}
        \left(\Sigma(A_1) - \Sigma(B_1)\right) +
        \left(\Sigma(A_2) - \Sigma(B_2)\right) +
        \left(\Sigma(A_3) - \Sigma(B_3)\right) +
        \left(\Sigma(A_4) - \Sigma(B_4)\right) = 0
    \end{displaymath}

    Observe, that by definition $A_i \cap B_i = \emptyset$ for all $i \in
    \{1,2,3,4\}$. So the problem reduces to finding 4 vectors $x_1,x_2,x_3,x_4
    \in \{-1,0,1\}^{n/4}$, such that:

    \begin{equation}
        \label{table-4-sum:eq}
        \overline{a}_1 \cdot \overline{x}_1 + \overline{a}_2 \cdot \overline{x}_2 + \overline{a}_3 \cdot
        \overline{x}_3 + \overline{a}_4 \cdot \overline{x}_4 = 0
        .
    \end{equation}

    because a term $\Sigma(A_i) - \Sigma(B_i)$ corresponds to $\overline{a}_i \cdot
    \overline{x}_i$ ($1$'s from $\overline{x}_i$ correspond to the elements of
    $A_i$ and $-1$'s from $\overline{x}_i$ correspond to the elements of $B_i$).

    Now, the algorithm is as follows. First enumerate all possible values of
    $\overline{a}_i \cdot \overline{x}_i$ for all $i \in \{1,2,3,4\}$ and store
    them in a table $T_i$. Along the way of value of $\overline{a}_i \cdot
    \overline{x}_i$ we store corresponding vector $\overline{x}$. Note, that
    $|T_i| = \Os(3^{n/4})$. Now run 4-SUM on input tables $T_i$ for $i \in
    \{1,2,3,4\}$, find $\overline{x}_i$ such that
    Equation~(\ref{table-4-sum:eq}) is satisfied. Then we find the corresponding
    sets $A_i$ and $B_i$ and return $(A_1\cup A_2 \cup A_3 \cup A_4, B_1\cup B_2
    \cup B_3 \cup B_4)$. The 4-SUM finds vectors that sum to $0$ from 4
    different input sets. Because we enumerated all possibilities the
    correctness follows.

    4-SUM runs in $\Ot(|I|^2)$ time and $\Ot(|I|)$ space where $|I|$ is
    the size of the input instance. In our case $|I| = \Os(3^{n/4})$ and the
    running time and space complexity of the algorithm follows.
\end{proof}

\section{Time-Space Tradeoff for \Ess}
\label{naive-tradeoff}

\citet{schroeppel} gave a time-space tradeoff for \Ss, such that $\mathcal{T}
\mathcal{S}^2 \le \Os(2^n)$ where $\mathcal{T}$ is a running time and
$\mathcal{S}$ is the space of the algorithm for \Ss and $\mathcal{S} \le
\Os(2^{n/4})$. In this section we observe that similar relation is true for
\Ess:

\begin{theorem}
    For all $\mathcal{S} \le \Os(3^{n/4})$, 
    \Ess can be solved in space $\mathcal{S}$ and time $\mathcal{T} \le
    \Os(\frac{3^n}{\mathcal{S}^2})$.
\end{theorem}

\begin{proof}
    Let $S$ be the input instance of \Ess and $\beta \in [0,1]$ be our trade-off
    parameter. By $A,B$ we will denote a solution to \Ess, i.e., $\Sigma(A) =
    \Sigma(B)$ and $A \cap B = \emptyset$. 
    
    Intuitively, for $\beta=1$ we will use a polyspace algorithm running in time
    $\Os(3^n)$ and for $\beta=0$ we will use a meet in the middle algorithm
    running in $\Os(3^{n/2}$ time and $\Os(3^{n/4})$ space. First we arbitrarily
    choose a set $X$ of $\beta n $ elements of $S$. Then we guess set $A\cap X$
    and set $B \cap X$. Finally we execute \Ess meet-in-the-middle algorithm for
    an instance $(S \setminus X) \cup \{\Sigma(A \cap X), \Sigma(B \cap X)\}$ of
    $n(1-\beta) + 2$ elements. The correctness follows because we checked all
    possible splits of $X$ into sets $A$ and $B$ and put them into the solution.
    We did not increase possible solutions hence if the answer to \Ess was NO
    then we will always answer NO. Similarly if the answer was YES, and the sets
    $A,B\subseteq S$ are a good solution, then for correctly guess $A\cap X$ and
    $B \cap X$ the constructed instance is a YES instance.
    
    The algorithm runs in time $\mathcal{T}(n,\beta) = \Os(3^{\beta n} \cdot
    \mathcal{T}(n(1-\beta))
    \le \Os(3^{\beta n} 3^{(1-\beta)n/2})$
    and space $\mathcal{S}(n,\beta) = \Os(\mathcal{S}((1-\beta)n) \le
    \Os(3^{(1-\beta)n/4})$ (see Appendix~\ref{table-4-sum-ess}). It
    follows that:
    \begin{displaymath}
        \mathcal{T}(n,\beta) \mathcal{S}(n,\beta)^2 \le \Os(3^n)
    \end{displaymath}
    Which gives us the final time-space tradeoff.

\end{proof}
 \section{Exact algorithm for Number Balancing}
\label{number-balancing}

Recall, that in the \emph{Number Balancing} problem you are given $n$ real numbers $a_1,\ldots,a_n
\in [0,1]$. The task is to find two disjoint subsets $I,J \subseteq [n]$, such that the
difference $|\sum_{i \in I} a_i - \sum_{j \in J} a_j |$ is minimized. In this
Section we show that our techniques transfer to the exact algorithm for Number
Balancing. To alleviate problems with the definition of the
computational model for real numbers, we will be solving the following problem:

\begin{definition}[Integer Number Balancing]
    In the \emph{Integer Number Balancing} problem, we are given a set $S$ of
    $n$ integers $a_1,\ldots,a_n \in \{0,\ldots,2^{\Oh(n)}\}$. The task is to
    find two disjoint subsets $I,J \subseteq [n]$, such that the difference
    $\left| \sum_{i \in I} a_i - \sum_{j \in J} a_j \right|$ is minimized.
\end{definition}

Note, that \citet{karmarkar82} defined Number Balancing for reals because they
were interested in approximation algorithms. For our purposes it is
convenient to assume that numbers are given as integers bounded by $2^{\Oh(n)}$.
For unbounded integers, some additional factors due to the arithmetic operations may occur.

\begin{theorem}
    \label{number-balancing-exact}
    Integer Number Balancing can be solved in $\Os(1.7088^n)$ time with high probability.
\end{theorem}

It is convenient to work with the following decision version of the problem:

\begin{definition}[Integer Number Balancing, decision version]
    \label{number-balancing-dec}

    In the decision version of \emph{Integer Number Balancing}, we are given a
    set $S$ of $n$ integers $a_1,\ldots,a_n \in \{0,\ldots,2^{\Oh(n)}\}$ and
    integer $\kappa$.  The task is decide if there exist two disjoint subsets
    $I,J \subseteq [n]$, such that $\left| \sum_{i \in I} a_i - \sum_{j \in J}
    a_j \right| \in [0,\kappa]$.

\end{definition}

The above decision version and minimization version are equivalent up to
polynomial factors: we use a binary search to for the smallest $\kappa$, for
which answer to the decision version of Integer Number Balancing is YES. The target
$\kappa \in [0, 2^{\Oh(n)}]$ so we need at most polynomial number of calls to the
oracle.

\subsection{Proof of Theorem~\ref{number-balancing-exact}}

First we observe, that our techniques also work for the generalization of \Ess.

\begin{definition}[Target \Ess problem]
    In the \emph{Target \Ess} problem, we are given a set $S$ of $n$ integers
    and integer $\kappa$. The task is to decide if there exist
    two disjoint nonempty subsets $A,B \subseteq S$, such that $\left|
    \Sigma(A) - \Sigma(B) \right| = \kappa$.
\end{definition}

\begin{theorem}
    \label{target-ess}
    \emph{Target \Ess} problem in $\Os(1.7088^n)$ time with high probability.
\end{theorem}

We give a sketch of the proof in Section~\ref{proof-target-ess}.

Now, we use an algorithm for Target \Ess to give an algorithm for
Integer Number Balancing. 
The observation is that decision version of Integer Number Balancing (see
Definition~\ref{number-balancing-dec}) asks if there exist two subsets $X,Y
\subseteq S$ such that $|\Sigma(X) - \Sigma(Y)| \in [0,\kappa]$. However
Theorem~\ref{target-ess} gives us an access to the oracle that determines if
there exist two subsets $X,Y \subseteq S$, such that $|\Sigma(X) - \Sigma(Y)| =
\kappa$. The following Lemma gives us tool for such a reduction:

\begin{lemma}[Shrinking Intervals, Theorem 1 from~\cite{nederlof-mfcs2012}]
    \label{intervals}
    Let $U$ be a set of cardinality $n$, let $\omega : U \rightarrow
    \{-W,\ldots,W\}$ be a weight function, and let $l < u$ be
    integers with $u-l > 1$. Then, there is a polynomial-time algorithm that
    returns a set of pairs $\Omega = \{ (\omega_1,v_1),\ldots,(\omega_T, v_T)\}$
    with $\omega_i : U \rightarrow \{-W,\ldots,W\}$ and integers $v_1,\ldots,v_T
    \in \{-W,\ldots,W\}$, such that:

    \begin{itemize}
        \item $T$ is at most $\Oh(n \log{(u-l)})$, and:
        \item for every set $X \subseteq U$ it holds that $\omega(X) \in [l,u]$
            if and only if there exist an index $i \in [T]$ such that $\omega_i(X) =
            v_i$.
    \end{itemize}
\end{lemma}

Note, that the corresponding Theorem in~\cite{nederlof-mfcs2012} was stated for
weight function $\omega : U \rightarrow \{0,\ldots,W\}$. However, the proof
in~\cite{nederlof-mfcs2012} does not need that assumption. For clarity, in
\cite{nederlof-mfcs2012} weight functions $\omega_i : U \rightarrow
\{-W,\ldots,W\}$ are of the following form: for set $X \subseteq U$ the function
is always $\omega_i(X) = \sum_{x \in X} w_x$ for some weights $w_i \in \mathbb{Z}$.

With Lemma~\ref{intervals} in hand we can now prove Theorem~\ref{number-balancing-exact}.

\begin{proof}[Proof of Theorem~\ref{number-balancing-exact}]
    Let $S$ be the set of $n$ integers $\{a_1,\ldots,a_n\}$ as in
    Definition~\ref{number-balancing-dec} and a target $\kappa$. Let $U = \{-n,\ldots,-1\} \cup
    \{1,\ldots,n\}$. For $z \in \mathbb{Z}$, let $\Sgn{z}$ be sign
    function, i.e., $\Sgn{z} = -1$ when $z < 0$, $\Sgn{0} = 0$ and $\Sgn{z} = 1$
    when $z > 0$. Moreover, for any $X \subseteq U$ let $\omega(X) = \sum_{x \in
    X} \Sgn{x} a_{|x|}$.

    We are given black-box access to the Theorem~\ref{target-ess}, i.e., for a given set
    $S'$ of integers we can decide if there exist two subsets $X,Y \subseteq S'$,
    such that $| \Sigma(X) - \Sigma(Y) | = \kappa$ in time $\Os(1.7088^n)$. We show
    that we can solve Integer Number Balancing by using polynomial number of calls to
    Theorem~\ref{target-ess}.

    First, observe that universe set $U$ and the weight function $\omega(X)$
    satisfy the conditions of Lemma~\ref{intervals}. Moreover, let $u = \kappa$ and
    $l = -\kappa$. Lemma~\ref{intervals} works in polynomial time and outputs
    pairs $P = \{(\omega_1,v_1),\ldots, (\omega_T, v_T)\}$. Now, the answer
    to the decision version of Integer Number Balancing on $S$ is YES iff there
    exists index $i \in [T]$ such that an answer to Target \Ess on instance
    $(\omega_i, v_i)$ is YES by Lemma~\ref{intervals}.

    For the running time observe, that the numbers are bounded by
    $2^{\Oh(n)}$, so $T = \Oh(\poly(n))$. Hence, we execute polynomial number of calls
    to Theorem~\ref{target-ess} and the running time follows.
\end{proof}

\subsection{Proof of Theorem~\ref{target-ess}}
\label{proof-target-ess}

What is left is to sketch that our techniques also apply to a more general
version of the problem.

We are given a set $S$ of $n$ integers and a target $\kappa$. We need to find
$X,Y \subseteq S$, such that $\Sigma(X) - \Sigma(Y) = \kappa$. First of all the
definition of \emph{minimum solution} for a target easily generalizes, i.e.,
we say that a solution $A,B \subseteq S$ such that $\Sigma(A) - \Sigma(B) =
\kappa$ is a minimum solution if its size $|A|+|B|$ is smallest possible.

Note, that the meet-in-the-middle algorithm for \Ess works for
Target \Ess (see Theorem~\ref{unbalanced-alg} and
Algorithm~\ref{alg:unbalanced}). The only difference is that in
Algorithm~\ref{alg:unbalanced}, we need to determine if there exist $x_1 \in
C_1, x_2 \in C_2$ such that $x_1 + x_2 = \kappa$.  The running time and analysis is exactly the same in that case.  

The main difference comes in the analysis of balanced case, i.e.,
Theorem~\ref{balanced-alg}. In that case we need to enumerate two sets $C_{t,p}$
and $C_{t-\kappa,p}$ (see Algorithm~\ref{alg:balanced-target})

\begin{algorithm}
    \caption{$\textsc{BalancedEqualSubsetSumTarget}(a_1,\ldots,a_n,\ell,\kappa)$}
	\label{alg:balanced-target}
\begin{algorithmic}[1]
    \State Pick a random prime $p$ in $[2^{n-\ell}, 2^{n-\ell+1}]$
    \State Pick a random number $t \in [1,2^{n-\ell}]$
    \State Let $C_1 = \{ X \subseteq S \; | \; \Sigma(X) \equiv_p t \}$ 
    \State Let $C_2 = \{ X \subseteq S \; | \; \Sigma(X) \equiv_p t-\kappa \}$ 
    \State Enumerate and store all elements of $C_1$ and $C_2$
    \Comment In time $\Os(|C_1| + |C_2| + 2^{n/2})$
    \State Find $X \in C_1$ and $Y \in C_2$, such that $\Sigma(X) - \Sigma(Y) = \kappa$
    \Comment In time $\Os(|C_1| + |C_2|)$
    \State \Return $(X\setminus Y, Y\setminus X)$
\end{algorithmic}
\end{algorithm}

What is left is to show, that Algorithm~\ref{alg:balanced-target} has the
running time $\Os(2^\ell)$ and finds a solution to Target \Ess with probability
$\Omega(1/\poly(n))$. The rest of the proof and analysis is exactly the same as
the proof of Theorem~\ref{balanced-alg}.

For the running time, note that $\Ex{|C_1|} \le \Os(2^{\ell})$ and $\Ex{|C_2|} \le
\Os(2^{\ell})$ because these sets are chosen in exactly the same way as set $C$
in Lemma~\ref{claim:sizeofl}. Moreover, we can enumerate both of them in time
$\Os(|C_1| + |C_2| + 2^{n/2}) \le \Os(2^\ell)$ by using Lemma~\ref{modenumeration} (recall that
$\ell > n/2$). Finally we can find $X \in C_1$ and $Y \in C_2$, such that
$\Sigma(X) - \Sigma(Y) = \kappa$ (if such $X,Y$ exist) in time $\Os(2^\ell)$ by
solving 2SUM. Hence, the running time of Algorithm~\ref{alg:balanced-target} is
$\Os(2^\ell)$.

For the correctness, observe that an analog of Lemma~\ref{obs:sizeofcandidates}
holds:

\begin{lemma}
    \label{obs:sizeofcandidates-target}
    Let $S$ be a set of $n$ positive integers with minimum solution size of $\ell$.
    Let

    \begin{equation}
        \Psi = \left\{\Sigma(X) \; | \; \exists Y \subseteq S \; \text{such
        that} \; X\neq Y \; \text{and} \; \Sigma(X)-\Sigma(Y) = \kappa \right\}
        .
    \end{equation}

    If $\ell > \frac{n}{2}$, then $|\Psi| \geq 2^{n-\ell}$ (note that all
    elements in $\Psi$ are different).
\end{lemma}

\begin{proof}[Proof of Lemma~\ref{obs:sizeofcandidates-target}]
    Similarly to the proof of Lemma~\ref{obs:sizeofcandidates-target} we assume, that
    there exist $A,B \subseteq S$, such that $A\cap B = \emptyset$, $|A| +
    |B| = \ell$ and $\Sigma(A) - \Sigma(B) = t$. Then we construct our set
    $\Psi$ be considering every subset $Z \subseteq S \setminus (A \cup B)$ and
    observing that:

    \begin{itemize}
        \item $\Sigma(A \cup Z) - \Sigma(B \cup Z) = \kappa$, and
        \item there are $2^{n-\ell}$ possible choices of set $Z$, and
        \item by the minimality of $A,B$ all sets $Z$ have a different sum.
    \end{itemize}
\end{proof}

And with that Lemma in hand we can prove the analogous to
Lemma~\ref{claim:prob}.

\begin{lemma}
    \label{claim:prob-target}

    Let $S$ be the set of $n$ positive integers bounded by $2^{\Oh(n)}$ with
    minimum solution $A,B$, $\Sigma(A) - \Sigma(B) = \kappa$ of size $\ell$ and
    $\ell > \frac{n}{2}$. For a random prime $p \in
    [2^{n-\ell},2^{n-\ell+1}]$ and a random $t \in [1,2^{n-\ell}]$ let $C_{t,p} = \{
        X \subseteq S \; | \; \Sigma(X) \equiv_p t \;\}$. Then,

    \begin{displaymath}
        \prob{ \exists X\subseteq C_{t,p}, Y \subseteq C_{t-\kappa,p} \; \Big| \; \Sigma(X) - \Sigma(Y) = \kappa, \; X \neq Y }{t,p}
        \ge \Omega(1/n^2).
    \end{displaymath}
\end{lemma}

\begin{proof}[Proof of Lemma~\ref{claim:prob-target}]
    Recall $\Psi$ from Lemma~\ref{obs:sizeofcandidates-target}. Note, that
    it is sufficient to show that there exist an element $a \in \Psi$, such that
    $a \equiv_p t$ with constant probability. Namely, if that is true, then $a \in C_{t,p}$ and by the
    definition of $\Psi$, there exists set $B \subseteq S$, such that $a -
    \Sigma(B) = \kappa$. Hence, $\Sigma(B) \in C_{t-\kappa,p}$ and the claim
    follows.

    The rest of the proof, i.e., showing that $a \equiv_p t$ with constant
    probability is analogous to the proof of Lemma~\ref{claim:prob}.
\end{proof}

With that in hand the correctness is analogous to the proof of correctness of
Theorem~\ref{balanced-alg}.
 
\end{appendices}

\end{document}